\documentclass[a4paper,conference]{IEEEtran}

\usepackage{amssymb}
\usepackage{amsfonts}
\usepackage[colorlinks,citecolor=blue,urlcolor=blue,breaklinks=true]{hyperref}
\usepackage{tikz}
\usepackage{ifthen}
\usepackage[cmex10]{amsmath} 

\begin{document}
\title{Constructing Optimal Authentication Codes\\ with Perfect Multi-fold Secrecy}

\author{%
\IEEEauthorblockN{Michael Huber}
\IEEEauthorblockA{University of Tuebingen\\
              Wilhelm Schickard Institute for Computer Science\\
              Sand~13, D-72076 Tuebingen, Germany\\
              Email: michael.huber@uni-tuebingen.de}
}

\newtheorem{corollary}{Corollary}
\newtheorem{theorem}{Theorem}
\newtheorem{lemma}{Lemma}
\newtheorem{problem}{Problem}
\newtheorem{example}{Example}
\newtheorem{remark}{Remark}

\maketitle

\begin{abstract}
We establish a construction of optimal authentication codes achieving perfect multi-fold secrecy by means of combinatorial designs. This continues the author's work~(ISIT~2009, cf.~\cite{Hub2009}) and answers an open question posed therein. As an application, we present the first infinite class of optimal codes that provide two-fold security against spoofing attacks and at the same time perfect two-fold secrecy.
\end{abstract}

\section{Introduction}\label{Intro}

Authentication and secrecy are two crucial concepts in cryptography and information security. Although independent in their nature, various scenarios require that both aspects hold simultaneously. For \emph{information-theoretic} or \emph{unconditional} security (i.e. robustness against an attacker that has unlimited computational resources), authentication and secrecy codes have been investigated for quite some time. The initial construction of authentication codes goes back to Gilbert, MacWilliams \& Sloane~\cite{gil74}. A more general and systematic theory of authentication was developed by Simmons (e.g.,~\cite{Sim85,Sim92}). Fundamental work on secrecy codes started with Shannon~\cite{Shan49}.

This paper deals with the construction of optimal authentication codes with perfect multi-fold secrecy. It continues the author's recent work~\cite{Hub2009}, which naturally extended results by Stinson~\cite{Stin90} on authentication codes with perfect secrecy. We will answer an important question left open in~\cite{Hub2009} that addresses the construction of authentication codes with perfect multi-fold secrecy for equiprobable source probability distributions. We establish a construction of optimal authentication codes which are multi-fold secure against spoofing attacks and simultaneously provide perfect multi-fold secrecy. This can be achieved by means of combinatorial designs. As an application, we present the first infinite class of optimal codes that achieve two-fold security against spoofing as well as perfect two-fold secrecy.

The paper is organized as follows: Necessary definitions and concepts from the theory of authentication and secrecy codes as well as from combinatorial design theory will be summarized in Section~\ref{Prelim}. Section~\ref{Auth} gives relevant combinatorial constructions of optimal authentication codes which bear no secrecy assumptions. In Section~\ref{Stin}, we review Stinson's constructions in~\cite{Stin90} and recent results from~\cite{Hub2009}. Section~\ref{New} is devoted to our new constructions.

\section{Preliminaries}\label{Prelim}

\subsection{Authentication and Secrecy Codes}\label{PrelimA}

\noindent We rely on the information-theoretical or unconditional secrecy model developed by Shannon~\cite{Shan49}, and by Simmons (e.g.,~\cite{Sim85,Sim92}) including authentication. Our notion complies, for the most part, with that of~\cite{Stin90,Mass86}. In this model of authentication and secrecy three participants are involved: a \emph{transmitter}, a \emph{receiver}, and an \emph{opponent}.  The transmitter wants to communicate information to the receiver via a public communications channel. The receiver in return would like to be confident that any received information actually came from the transmitter and not from some opponent (\emph{integrity} of information). The transmitter and the receiver are assumed to trust each other. Sometimes this is also called an \emph{$A$-code}.

In what follows, let $\mathcal{S}$ denote a set of $k$ \emph{source states} (or \emph{plaintexts}), $\mathcal{M}$ a set of $v$ \emph{messages} (or \emph{ciphertexts}), and $\mathcal{E}$ a set of $b$ \emph{encoding rules} (or \emph{keys}). Using an encoding rule $e\in \mathcal{E}$,
the transmitter encrypts a source state $s \in \mathcal{S}$ to obtain the message $m=e(s)$ to be sent over the channel. The encoding rule is an injective function from $\mathcal{S}$ to $\mathcal{M}$, and is communicated to the receiver via a secure channel prior to any messages being sent.
For a given encoding rule $e \in \mathcal{E}$, let $M(e):=\{e(s) : s \in \mathcal{S}\}$ denote the set of \emph{valid} messages. For an encoding rule $e$ and a set $M^* \subseteq M(e)$ of distinct messages, we define $f_e(M^*):=\{s \in \mathcal{S}: e(s) \in M^*\}$,
i.e., the set of source states that will be encoded under encoding rule $e$ by a message in $M^*$.
A received message $m$ will be accepted by the receiver as being authentic if and only if $m \in M(e)$. When this is fulfilled, the receiver decrypts the message $m$ by applying the decoding rule $e^{-1}$, where \[e^{-1}(m)=s \Leftrightarrow e(s)=m.\]
An authentication code can be represented algebraically by a $(b \times k)$\emph{-encoding matrix} with the rows indexed by the encoding rules, the columns indexed by the source states, and the entries defined by $a_{es}:=e(s)$ ($1\leq e \leq b$, $1\leq s \leq k$).

We address the scenario of a \emph{spoofing attack} of order $i$ (cf.~\cite{Mass86}):
Suppose that an opponent observes $i\geq 0$ distinct messages, which are sent through the public channel using the same encoding rule. The opponent then inserts a new message $m^\prime$ (being distinct from the $i$ messages already sent), hoping to have it accepted by the receiver as authentic.
The cases $i=0$ and $i=1$ are called \emph{impersonation game} and \emph{substitution game}, respectively. These cases have been studied in detail in recent years (e.g.,~\cite{Stin92,Stin96}), however less is known for the cases $i \geq 2$. In this article, we focus on those cases where $i \geq 2$.

For any $i$, we assume that there is some probability distribution on the set of \mbox{$i$-subsets} of source states, so that any set of $i$ source states has a non-zero probability of occurring. For simplification, we ignore the order in which the $i$ source states occur, and assume that no source state occurs more than once.
Given this probability distribution  $p_S$  on $\mathcal{S}$, the receiver and transmitter choose a probability distribution $p_E$ on $\mathcal{E}$  (called \emph{encoding strategy}) with associated independent random variables $S$ and $E$, respectively. These distributions are known to all participants and induce a third distribution, $p_M$, on $\mathcal{M}$ with associated random variable $M$. The \emph{deception probability} $P_{d_i}$ is the probability that the opponent can deceive the receiver with a spoofing attack of order $i$. The following theorem (cf.~\cite{Mass86}) provides combinatorial lower bounds.

\begin{theorem}$[$Massey$]$
In an authentication code with $k$ source states and $v$ messages, the deception probabilities are bounded below by
\[P_{d_i}\geq \frac{k-i}{v-i}.\]
\end{theorem}

An authentication code is called $t_A$\emph{-fold secure against spoofing} if
$P_{d_i}= (k-i)/(v-i)$ for all $0 \leq i \leq t_A$.

Moreover, we consider the concept of perfect multi-fold secrecy which has been introduced by Stinson~\cite{Stin90} and generalizes Shannon's fundamental idea of perfect (one-fold) secrecy (cf.~\cite{Shan49}). We say that an authentication code has \emph{perfect} $t_S$\emph{-fold secrecy} if, for every positive integer $t^* \leq t_S$, for every set $M^*$ of $t^*$ messages observed in the channel, and for every set $S^*$ of $t^*$ source states, we have \[p_S(S^*|M^*)=p_S(S^*).\]
That is, the \emph{a posteriori} probability distribution on the $t^*$ source states, given that a set of $t^*$ messages is observed, is identical to
the \emph{a priori} probability distribution on the $t^*$ source states.

When clear from the context, we often only write $t$ instead of $t_A$ resp. $t_S$.

\subsection{Combinatorial Designs}

\noindent We recall the definition of a combinatorial $t$-design.
For positive integers $t \leq k \leq v$ and $\lambda$, a \mbox{\emph{$t$-$(v,k,\lambda)$ design}} $\mathcal{D}$ is a pair \mbox{$(X,\mathcal{B})$}, satisfying the following properties:

\begin{enumerate}
\item[(i)] $X$ is a set of $v$ elements, called \emph{points},

\item[(ii)] $\mathcal{B}$ is a family of \mbox{$k$-subsets} of $X$, called \emph{blocks},

\item[(iii)] every \mbox{$t$-subset} of $X$ is contained in exactly $\lambda$ blocks.

\end{enumerate}
We denote points by lower-case and blocks by upper-case Latin
letters.
Via convention, let $b:=\left| \mathcal{B} \right|$ denote the number of blocks.
Throughout this article, `repeated
blocks' are not allowed, that is, the same \mbox{$k$-subset}
of points may not occur twice as a block. If $t<k<v$ holds, then we
speak of a \emph{non-trivial} \mbox{$t$-design}.
For historical reasons, a \mbox{$t$-$(v,k,\lambda)$ design} with
$\lambda =1$ is called a \emph{Steiner \mbox{$t$-design}} (sometimes
also a \emph{Steiner system}).
The special case of a Steiner design with parameters $t=2$ and $k=3$
is called a \emph{Steiner triple system $\mbox{STS}(v)$} of order $v$. A Steiner design
with parameters $t=3$ and $k=4$ is called a \emph{Steiner quadruple system $\mbox{SQS}(v)$} of order $v$.
Specifically, we are interested in Steiner quadruple systems in this paper.
As a simple example, the vector space
$\mathbb{Z}_2^d$ ($d \geq 3$) with the set $\mathcal{B}$ of blocks taken to be the set of all subsets of four distinct elements of
$\mathbb{Z}_2^d$ whose vector sum is zero, is a non-trivial \emph{boolean} Steiner quadruple system $\mbox{SQS}(2^d)$. More geometrically, these $\mbox{SQS}(2^d)$ consist of the points and planes of the \mbox{$d$-dimensional} binary affine space $AG(d,2)$.

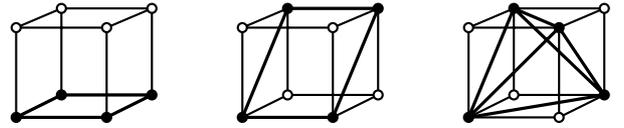
\begin{figure}[h]

\bigskip

\centering

\begin{tikzpicture}[scale=0.85,thick]

\filldraw [draw=black!100,fill=black!100]

(0,0) circle (2pt)

(1.4,0) circle (2pt)

(0.7,0.35) circle (2pt)

(2.1,0.35) circle (2pt);

\draw

(0,0) -- (0,1.4)

(0,0) -- (0.7,0.35)

(1.4,0) -- (1.4,1.4)

(1.4,0) -- (2.1,0.35)

(0.7,0.35) -- (2.1,0.35)

(2.1,0.35) -- (2.1,1.7)

(0,1.4) -- (1.4,1.4)

(0.7,1.7) -- (2.1,1.7)

(1.4,1.4) -- (2.1,1.7)

(0,1.4) -- (0.7,1.7)

(0.7,0.35) -- (0.7,1.7);

\filldraw [draw=black!100,fill=black!0]

(0,1.4) circle(2pt)

(1.4,1.4) circle (2pt)

(0.7,1.7) circle(2pt)

(2.1,1.7) circle (2pt);

\draw[very thick]

(0,0) -- (1.4,0)

(0,0) -- (0.7,0.35)

(1.4,0) -- (2.1,0.35)

(0.7,0.35) -- (2.1,0.35);

\filldraw [draw=black!100,fill=black!100]

(3.5,0) circle (2pt)

(4.9,0) circle (2pt)

(4.2,1.7) circle(2pt)

(5.6,1.7) circle (2pt);

\draw (3.5,0) -- (4.9,0)

(3.5,0) -- (3.5,1.4)

(3.5,0) -- (4.2,0.35)

(4.9,0) -- (4.9,1.4)

(4.9,0) -- (5.6,0.35)

(4.2,0.35) -- (5.6,0.35)

(5.6,0.35) -- (5.6,1.7)

(3.5,1.4) -- (4.9,1.4)

(4.2,1.7) -- (5.6,1.7)

(4.9,1.4) -- (5.6,1.7)

(3.5,1.4) -- (4.2,1.7)

(4.2,0.35) -- (4.2,1.7);

\filldraw [draw=black!100,fill=black!0]

(4.2,0.35) circle (2pt)

(5.6,0.35) circle (2pt)

(3.5,1.4) circle(2pt)

(4.9,1.4) circle (2pt);

\draw[very thick]

(3.5,0) -- (4.9,0)

(3.5,0) -- (4.2,1.7)

(4.9,0) -- (5.6,1.7)

(4.2,1.7) -- (5.6,1.7);

\filldraw [draw=black!100,fill=black!100]

(7,0) circle (2pt)

(7.7,1.7) circle(2pt)

(8.4,1.4) circle (2pt)

(9.1,0.35) circle (2pt);

\draw (7,0) -- (8.4,0)

(7,0) -- (7,1.4)

(7,0) -- (7.7,0.35)

(8.4,0) -- (8.4,1.4)

(8.4,0) -- (9.1,0.35)

(7.7,0.35) -- (9.1,0.35)

(9.1,0.35) -- (9.1,1.7)

(7,1.4) -- (8.4,1.4)

(7.7,1.7) -- (9.1,1.7)

(8.4,1.4) -- (9.1,1.7)

(7,1.4) -- (7.7,1.7)

(7.7,0.35) -- (7.7,1.7);

\filldraw [draw=black!100,fill=black!00]

(8.4,0) circle (2pt)

(7.7,0.35) circle (2pt)

(7,1.4) circle(2pt)

(9.1,1.7) circle (2pt);

\draw[very thick]

(7,0) -- (7.7,1.7)

(7,0) -- (8.4,1.4)

(7,0) -- (9.1,0.35)

(7.7,1.7) -- (9.1,0.35)

(8.4,1.4) -- (9.1,0.35)

(8.4,1.4) -- (7.7,1.7);
\end{tikzpicture}

\caption{\small{Illustration of the unique $\mbox{SQS}(8)$, with three types of blocks: faces, opposite edges, and inscribed regular tetrahedra.}}\label{cube}

\end{figure}

For the existence of \mbox{$t$-designs}, basic necessary
conditions can be obtained via elementary counting arguments (see,
for instance,~\cite{BJL1999}):

\begin{lemma}\label{s-design}
Let $\mathcal{D}=(X,\mathcal{B})$ be a \mbox{$t$-$(v,k,\lambda)$}
design, and for a positive integer $s \leq t$, let $S \subseteq X$
with $\left|S\right|=s$. Then the number of blocks containing
each element of $S$ is given by
\[\lambda_s = \lambda \frac{{v-s \choose t-s}}{{k-s \choose t-s}}.\]
In particular, for $t\geq 2$, a \mbox{$t$-$(v,k,\lambda)$} design is
also an \mbox{$s$-$(v,k,\lambda_s)$} design.
\end{lemma}

It is customary to set $r:= \lambda_1$ denoting the
number of blocks containing a given point. It follows

\begin{lemma}\label{Comb_t=5}
Let $\mathcal{D}=(X,\mathcal{B})$ be a \mbox{$t$-$(v,k,\lambda)$}
design. Then the following holds:
\begin{enumerate}

\item[{(a)}] $bk = vr.$

\smallskip

\item[{(b)}] $\displaystyle{{v \choose t} \lambda = b {k \choose t}.}$

\smallskip

\item[{(c)}] $r(k-1)=\lambda_2(v-1)$ for $t \geq 2$.

\end{enumerate}
\end{lemma}

For encyclopedic accounts of key results in design theory, we refer to~\cite{BJL1999,crc06}. Various connections of designs with coding and information theory can be found in a recent survey~\cite{Hu2009} (with many additional references therein).

\section{Optimal Authentication Codes}\label{Auth}

For our further purposes, we summarize the state-of-the-art for authentication codes which bear no secrecy assumptions.
The following theorem (cf.~\cite{Mass86,Sch86}) gives a combinatorial lower bound on the number of encoding rules.

\begin{theorem}$[$Massey--Sch\"{o}bi$]$
If an authentication code is $(t-1)$-fold against spoofing, then the number of encoding rules is bounded below by
\[b \geq \frac{{v \choose t}}{{k \choose t}}.\]
\end{theorem}

An authentication code is called \emph{optimal} if the number of encoding rules meets the lower bound with equality. When the source states are known to be independent and equiprobable, optimal authentication codes which are $(t-1)$-fold secure against spoofing can be constructed via \mbox{$t$-designs} (cf.~\cite{Stin90,Sch86,DeS88}).

\begin{theorem}$[$DeSoete--Sch\"{o}bi--Stinson$]$\label{general}
Suppose there is a \mbox{$t$-$(v,k,\lambda)$} design. Then there is an authentication code for $k$ equiprobable source states, having $v$ messages and $\lambda \cdot {v \choose t}/{k \choose t}$ encoding rules, that is $(t-1)$-fold secure against spoofing. Conversely, if there is an authentication code for $k$ equiprobable source states, having  $v$ messages and ${v \choose t}/{k \choose t}$ encoding rules, that is $(t-1)$-fold secure against spoofing, then there is a Steiner \mbox{$t$-$(v,k,1)$} design.
\end{theorem}

\section{Stinson's Constructions \& Recent Results}\label{Stin}

Using the notation introduced in Section~\ref{PrelimA}, we review in Tables~\ref{equi1} and~\ref{equi2} previous constructions from~\cite{Stin90,Hub2009} for equiprobable source probability distributions. This lists all presently known optimal authentication codes with perfect secrecy.

\begin{table}
\renewcommand{\arraystretch}{1.3}
\caption{Optimal authentication codes with perfect secrecy:\newline Infinite classes}\label{equi1}
\begin{center}
\begin{tabular}{|cc||ccc|c|}
  \hline
  $t_A$ & $t_S$ & $k$ & $v$ & $b$ & \mbox{Ref.} \\
  \hline \hline 
      1 & 1  & $q+1$ & $\frac{q^{d+1}-1}{q-1}$ & $\frac{v(v-1)}{k(k-1)}$ & \cite{Stin90}\\
        &    & $q$  {\scriptsize{\mbox{prime power}}} &  $d \geq 2$ {\scriptsize{\mbox{even}}} & & \\
   \hline
1 & 1  & $3$ & $v \equiv 1$ ({\scriptsize{\mbox{mod}}} $6$) & $\frac{v(v-1)}{6}$ & \cite{Hub2009}\\
   \hline
1 & 1  & $4$ & $v \equiv 1$ ({\scriptsize{\mbox{mod}}} $12$) & $\frac{v(v-1)}{12}$ & \cite{Hub2009}\\
   \hline   
1 & 1  & $5$ & $v \equiv 1$ ({\scriptsize{\mbox{mod}}} $20$) & $\frac{v(v-1)}{20}$ & \cite{Hub2009}\\
   \hline \hline
2 & 1  & $q+1$ & $q^d+1$ & $\frac{v(v-1)(v-2)}{k(k-1)(k-2)}$ & \cite{Hub2009}\\
        &    & $q$ {\scriptsize{\mbox{prime power}}} &  $d \geq 2$ {\scriptsize{\mbox{even}}} & & \\
   \hline
2 & 1  & $4$ & $v \equiv 2, 10$ ({\scriptsize{\mbox{mod}}} $24$) & $\frac{v(v-1)(v-2)}{24}$ & \cite{Hub2009}\\
   \hline
\end{tabular}
\end{center}
\end{table}

\begin{table}
\renewcommand{\arraystretch}{1.3}
\caption{Optimal authentication codes with perfect secrecy:\newline Further examples}\label{equi2}
\begin{center}
\begin{tabular}{|c c||c c c |c|}
  \hline
  $t_A$ & $t_S$ & $k$ & $v$ & $b$ & \mbox{Ref.} \\
  \hline \hline
2 & 1 & 5  & 26 & 260 & \cite{Hub2009}\\
\hline \hline
 &  & 5  & 11 & 66 & \cite{Hub2009}\\
 &  & 7  & 23 & 253 & \cite{Hub2009}\\
 &  & 5  & 23 & 1.771 & \cite{Hub2009}\\
 &  & 5  & 47 & 35.673 & \cite{Hub2009}\\
3 & 1 & 5  & 83 & 367.524 & \cite{Hub2009}\\
 &    & 5  & 71 & 194.327 & \cite{Hub2009}\\
 &    & 5  & 107 & 1.032.122 &\cite{Hub2009}\\
 &    & 5  & 131 & 2.343.328 &\cite{Hub2009}\\
 &   & 5  & 167 & 6.251.311 &\cite{Hub2009}\\
 &    & 5  & 243 & 28.344.492 &\cite{Hub2009}\\
   \hline \hline
  & & 6  & 12 & 132 & \cite{Hub2009}\\
4  & 1& 6  & 84 & 5.145.336 &\cite{Hub2009}\\
  & & 6  & 244 & 1.152.676.008 &\cite{Hub2009}\\
  \hline
\end{tabular}
\end{center}
\end{table}

\section{New Constructions}\label{New}

Starting from the condition of perfect \mbox{$t$-fold} secrecy, we obtain via Bayes' Theorem that
\begin{align*}{}
p_S(S^*|M^*) &=  \frac{p_M(M^*|S^*) p_S(S^*)}{p_M(M^*)}\\ &=
\frac{\sum_{\{e \in \mathcal{E}: S^*=f_e(M^*)\}} p_E(e)p_S(S^*)}{\sum_{\{e \in \mathcal{E}:M^* \subseteq M(e)\}} p_E(e)p_S(f_e(M^*))} = p_S(S^*).
\end{align*}

It follows

\begin{lemma}\label{frequency_multi}
An authentication code has perfect $t$-fold secrecy if and only if, for every positive integer $t^* \leq t$, for every set $M^*$ of $t^*$ messages observed in the channel and for every set $S^*$ of $t^*$ source states, we have
\[\sum_{\{e \in \mathcal{E}: S^*=f_e(M^*)\}} p_E(e)=\sum_{\{e \in \mathcal{E}: M^* \subseteq M(e)\}} p_E(e)p_S(f_e(M^*)).\]
\end{lemma}

Hence, if the encoding rules in a code are used with equal probability, then for every $t^* \leq t$, a given set of $t^*$ messages occurs with the same frequency in each $t^*$ columns of the encoding matrix.

We can now establish an extension of the main theorem in~\cite{Hub2009}. Our construction yields optimal authentication codes which are multi-fold secure against spoofing and provide perfect multi-fold secrecy.

\begin{theorem}\label{mythm_IZS}
Suppose there is a Steiner \mbox{$t$-$(v,k,1)$} design, where ${v \choose t^*}$ divides the number of blocks $b$ for every positive integer  $t^* \leq t-1$. Then there is an optimal authentication code for $k$ equiprobable source states, having $v$ messages and ${v \choose t}/{k \choose t}$ encoding rules, that is \mbox{($t-1$)}-fold secure against spoofing and simultaneously provides perfect \mbox{$(t-1)$}-fold secrecy.
\end{theorem}

\begin{proof}
Let $\mathcal{D}=(X,\mathcal{B})$ be a Steiner \mbox{$t$-$(v,k,1)$} design, where ${v \choose t^*}$ divides $b$ for every positive integer $t^* \leq t-1$. By Theorem~\ref{general}, the authentication code has \mbox{$(t-1)$}-fold security against spoofing attacks. Hence, it remains to prove that the code also achieves perfect \mbox{$(t-1)$}-fold secrecy under the assumption that the encoding rules are used with equal probability. With respect to Lemma~\ref{frequency_multi}, we have to show that, for every $t^*\leq t-1$, a given set of $t^*$ messages occurs with the same frequency in each $t^*$ columns of the resulting encoding matrix. This can be accomplished by ordering, for each $t^*\leq t-1$, every block of $\mathcal{D}$ in such a way that  every \mbox{$t^*$-subset} of $X$ occurs in each possible choice in precisely $b/{v \choose t^*}$ blocks.
Since every $t^*$-subset of $X$ occurs in exactly $\lambda_{t^*}={v-{t^*} \choose t-{t^*}}/{k-{t^*} \choose t-{t^*}}$ blocks due to Lemma~\ref{s-design}, necessarily ${k \choose {t^*}}$ must divide $\lambda_{t^*}$. By Lemma~\ref{Comb_t=5}~(b), this is equivalent to saying that ${v \choose t^*}$ divides $b$.
To show that the condition is also sufficient, we consider the bipartite (\mbox{$t^*$-subset}, block) incidence graph of $\mathcal{D}$ with vertex set ${X \choose t^*} \cup \mathcal{B}$, where $(\{x_i\}_{i=1}^{t^*},B)$ is an edge if and only if $x_i \in B$ ($1 \leq i \leq t^*$) for $\{x_i\}_{i=1}^{t^*} \in {X \choose t^*}$ and $B \in \mathcal{B}$. An ordering on each block of $\mathcal{D}$ can be obtained via an edge-coloring of this graph using ${k \choose t^*}$ colors in such a way that each vertex $B \in \mathcal{B}$ is adjacent to one edge of each color, and each vertex $\{x_i\}_{i=1}^{t^*} \in {X \choose t^*}$ is adjacent to $b/ {k \choose t^*}$ edges of each color. Specifically, this can be done by first splitting up each vertex $\{x_i\}_{i=1}^{t^*}$ into $b/{k \choose t^*}$ copies, each having degree ${k \choose t^*}$, and then by finding an appropriate edge-coloring of the resulting ${k \choose t^*}$-regular bipartite graph using ${k \choose t^*}$ colors. The claim follows now by taking the ordered blocks as encoding rules, each used with equal probability.
\end{proof}

\begin{remark}
It follows from the proof that we may obtain optimal authentication codes that provide \mbox{($t-1$)}-fold security against spoofing and at the same time perfect \mbox{$(t^\prime-1)$}-fold secrecy for $t^\prime \leq t$, when the assumption of the above theorem holds with ${v \choose t^*}$ divides $b$ for every positive integer $t^* \leq t^\prime-1$.
\end{remark}

As an application, we give an infinite class of optimal codes which are two-fold secure against spoofing and achieve perfect two-fold secrecy. This appears to be the first infinite class of authentication and secrecy codes with these properties.

\begin{theorem}\label{mythm2_IZS}
For all positive integers $v \equiv 2$ (mod $24$), there is an optimal authentication code for $k=4$ equiprobable source states, having $v$ messages, and $v(v-1)(v-2)/24$ encoding rules, that is two-fold secure against spoofing and provides perfect two-fold secrecy.
\end{theorem}

\begin{proof}
We will make use of Steiner quadruple systems (cf.~Section~\ref{PrelimA}). Hanani~\cite{Han1960} showed that a necessary and sufficient condition for the existence of a $\mbox{SQS}(v)$ is that $v \equiv 2$ or $4$ (mod $6$) $(v \geq 4)$. Hence, the condition $v \mid b$ is fulfilled when $v \equiv 2$ or $10$ (mod $24$) and the condition ${v \choose 2} \mid b$ when $v \equiv 2$ (mod $12$) in view  Lemma~\ref{Comb_t=5}~(b). Therefore, if we assume that $v \equiv 2$ (mod $24$), then we can apply Theorem~\ref{mythm_IZS} to establish the claim.
\end{proof}

We present the smallest example:

\begin{example}
An optimal authentication code for $k=4$ equiprobable source states, having $v=26$ messages, and $b=650$ encoding rules, that is two-fold secure against spoofing and provides perfect two-fold secrecy can be constructed from a Steiner quadruple system $\mbox{SQS}(26)$. Each encoding rule is used with probability $1/650$.
\end{example}

\begin{remark}
For $v=26$, the first $\mbox{SQS}(v)$ was constructed by Fitting~\cite{Fitt1915}, admitting a \mbox{$v$-cycle} as an automorphism (\emph{cyclic} $\mbox{SQS}(v)$). We generally remark that the number $N(v)$ of non-isomorphic $\mbox{SQS}(v)$ is only known for $v=8,10,14,16$ with $N(8)=N(10)=1$, $N(14)=4$, and $N(16)=1{,}054{,}163$ (cf.~\cite{KOP2006}). Lenz~\cite{Lenz1985} proved that for the admissible values of $v$, the number $N(v)$ grows exponentially, i.e. $\liminf_{v \rightarrow \infty} \frac{\log N(v)}{v^3}>0$. For comprehensive survey articles on Steiner quadruple systems, we refer the reader to~\cite{HartPh1992,LindRo1978}. For classifications of specific classes of highly regular Steiner quadruple systems and Steiner designs, see, e.g.,~\cite{Hu2010,Hu2008}.
\end{remark}

\section*{Acknowledgment}

The author thanks Doug Stinson for an interesting conversation on this topic. The author gratefully acknowledges support of his work by the Deutsche Forschungsgemeinschaft (DFG) via a Heisenberg grant (Hu954/4) and a Heinz Maier-Leibnitz Prize grant (Hu954/5).


\begin{thebibliography}{22}

\bibitem{Hub2009}
M. Huber, ``Authentication and secrecy codes for equiprobable source probability distributions'', in \emph{Proc. IEEE International Symposium on Information Theory (ISIT) 2009}, pp.~1105--1109, 2009.

\bibitem{gil74}
E.~N. Gilbert, F.~J. MacWilliams and N.~J.~A. Sloane, ``Codes which detect deception'',
\emph{Bell Syst.\ Tech.\ J.}, vol.\ 53, pp.~405--424, 1974.

\bibitem{Sim85}
G.~J. Simmons, ``Authentication theory/coding theory'', in
\emph{Advances in Cryptology -- CRYPTO '84}, ed. by G.~R. Blakley and D.~Chaum, Lecture Notes in Computer Science, vol.\ 196, Springer, Berlin,
Heidelberg, New York, pp.~411--432, 1985.

\bibitem{Sim92}
G.~J. Simmons, ``A survey of information authentication'', in
\emph{Contemporary Cryptology: The Science of Information Integrity}, ed. by G.~J. Simmons, IEEE Press, Piscataway, pp.~379--419, 1992.

\bibitem{Shan49}
C. E. Shannon, ``Communication theory of secrecy systems'',
\emph{Bell Syst.\ Tech.\ J.}, vol.\ 28, pp.~656--715, 1949.

\bibitem{Stin90}
D.~R. Stinson, ``The combinatorics of authentication and secrecy codes'',
\emph{J.\ Cryptology}, vol.\ 2, pp.~23--49, 1990.

\bibitem{Mass86}
J.~L. Massey, ``Cryptography -- a selective survey'', in
\emph{Digital\ Communications}, ed. by E.~Biglieri and G.~Prati, North-Holland, Amsterdam, New York, Oxford, pp.~3--21, 1986.

\bibitem{Stin92}
D.~R. Stinson, ``Combinatorial characterizations of authentication codes'',
\emph{Designs, Codes and Cryptography}, vol.\ 2, pp.~175--187, 1992.

\bibitem{Stin96}
D.~R. Stinson and R.~S. Rees, ``Combinatorial characterizations of authentication codes II'',
\emph{Designs, Codes and Cryptography}, vol.\ 7, pp.~239--259, 1996.

\bibitem{BJL1999}
Th. Beth, D.~Jungnickel and H.~Lenz, \emph{Design {Theory}}, vol. {I} and
  {II}, Encyclopedia of Math. and Its Applications, vol.\ 69/78, Cambridge Univ.
  Press, Cambridge, 1999.

\bibitem{crc06}
C.~J. Colbourn and J.~H. Dinitz (eds.), \emph{Handbook of {Combinatorial}
  {Designs}}, 2nd ed., CRC Press, Boca Raton, 2006.

\bibitem{Hu2009}
M. Huber, ``Coding theory and algebraic combinatorics'', in
\emph{Selected Topics in Information and Coding Theory}, ed. by I. Woungang \emph{et al.}, World Scientific, Singapore, 38 pages, 2010 (in press). Preprint at arXiv:0811.1254v1.

\bibitem{Sch86}
P. Sch\"{o}bi, ``Perfect authentication systems for data sources with arbitrary statistics'' (presented at EUROCRYPT '86), unpublished.

\bibitem{DeS88}
M. De Soete, ``Some constructions for authentication - secrecy codes'', in \emph{Advances in Cryptology -- EUROCRYPT '88}, ed. by Ch. G. G\"{u}nther, Lecture Notes in Computer Science, vol.\ 330, Springer, Berlin,
Heidelberg, New York, pp.~23--49, 1988.

\bibitem{Han1960}
H.~Hanani, ``On quadruple systems'', \emph{Canad. J. Math.}, vol.\ 12,
pp.~145--157, 1960.

\bibitem{Fitt1915}
F. Fitting, ``Zyklische L{\"o}sungen des Steiner'schen Problems'',
\emph{Nieuw. Arch. Wisk.}, vol.\ 11, pp.~140--148, 1915.

\bibitem{KOP2006}
P.~Kaski, P.~R.~J. \"{O}sterg{\aa}rd and O.~Pottonen, ``The {Steiner} quadruple systems of order $16$'',
\emph{J. Combin. Theory, Series A}, vol.\ 113, pp.~1764--1770, 2006.

\bibitem{Lenz1985}
H. Lenz, ``On the number of {Steiner} quadruple systems'',
\emph{Mitt. Math. Sem. Giessen}, vol.\ 169, pp.~55--71, 1985.

\bibitem{HartPh1992}
A.~Hartman and K.~T. Phelps, ``Steiner quadruple systems'', in:
\emph{Contemporary Design Theory}, ed. by J. H. Dinitz and D. R. Stinson,
Wiley, New York, pp.~205--240, 1992.

\bibitem{LindRo1978}
C.~C. Lindner and A.~Rosa, ``Steiner quadruple systems -- {A} survey'',
\emph{Discrete Math.}, vol.\ 22, pp.~147--181, 1978.

\bibitem{Hu2010}
M. Huber, ``Almost simple groups with socle $L_n(q)$ acting on Steiner quadruple systems'', \emph{J. Combin. Theory, Series A}, 4 pages, 2010 (in press). Preprint at arXiv:0907.1281v1.

\bibitem{Hu2008}
M.~Huber, \emph{Flag-transitive {S}teiner {D}esigns}, {Birkh\"{a}user}, Basel,
Berlin, Boston, 2009.

\end{thebibliography}
\end{document}